\documentclass[11pt,letterpaper]{article}

\usepackage[margin=1in]{geometry}
\usepackage{comment,algorithm,algorithmic,multicol}

\usepackage{color,colortbl}
\usepackage{float}
\usepackage{amsthm}
\usepackage{amsmath}
\usepackage{amssymb}
\usepackage{graphicx}
\usepackage{xspace}
\usepackage{nicefrac}
\usepackage{tablefootnote}
\usepackage[colorinlistoftodos,prependcaption,textsize=tiny]{todonotes}

\usepackage{thmtools}

\definecolor{Darkblue}{rgb}{0,0,0.4}
\definecolor{Brown}{cmyk}{0,0.61,1.,0.60}
\definecolor{Purple}{cmyk}{0.45,0.86,0,0}

\usepackage[colorlinks,allcolors=blue,pagebackref]{hyperref}
\usepackage[nameinlink]{cleveref}

\newcommand{\etal}{{\em et al. \xspace}}

\numberwithin{equation}{section} 

\newtheorem{theorem}{Theorem}[section]
\newtheorem{lemma}[theorem]{Lemma}
\newtheorem{corollary}[theorem]{Corollary}

\newtheorem{claim}[theorem]{Claim}

\newtheorem{question}[theorem]{Question}

\newcommand{\namedref}[2]{\hyperref[#2]{#1~\ref*{#2}}}

\newcommand{\N}{\mathbb{N}}
\newcommand{\R}{\mathbb{R}}

\newcommand{\polylog}{{\rm polylog}}

\newcommand{\eps}{\epsilon}

\newcommand{\theoremref}[1]{\namedref{Thm.}{#1}}
\newcommand{\corollaryref}[1]{\namedref{Cor.}{#1}}

\usepackage{pdfsync}
\usepackage{authblk}

\title{Labelings vs. Embeddings:\\ On Distributed and Prioritized Representations of Distances
\thanks{A preliminary version of this paper appeared in the proceedings of SODA 2020 \cite{FGK20}. This version contains the proofs of Theorems \ref{thm:PlanarEmbeddingLB} and \ref{thm:PlanarPrioritzedLB}, omitted from the preliminary version.}}
\author[1]{Arnold Filtser\thanks{This research was supported by the Israel science foundation  (grant No. 1042/22).}}
\author[2]{Lee-Ad Gottlieb\thanks{Work partially supported by ISF grant 1602/19.}}
\author[3]{Robert Krauthgamer\thanks{Work partially supported by ONR Award N00014-18-1-2364, the Israel Science Foundation grant \#1086/18, and a Minerva Foundation grant.
  }}

\affil[1]{Bar Ilan University\\ \texttt{arnold.filtser@biu.ac.il}}
\affil[2]{Ariel University\\ \texttt{leead@ariel.ac.il}}
\affil[3]{Weizmann Institute of Science\\ \texttt{robert.krauthgamer@weizmann.ac.il}}

\date{}

\begin{document}
\maketitle
\setcounter{page}{1}
\begin{abstract}
We investigate for which metric spaces the performance of distance labeling and of $\ell_\infty$-embeddings differ, and how significant can this difference be. Recall that a distance labeling is a distributed representation of distances in a metric space $(X,d)$, where each point $x\in X$ is assigned a succinct label, such that the distance between any two points $x,y \in X$ can be approximated given only their labels. A highly structured special case is an embedding into $\ell_\infty$, where each point $x\in X$ is assigned a vector $f(x)$
such that $\|f(x)-f(y)\|_\infty$ is approximately $d(x,y)$. The performance of a distance labeling or an $\ell_\infty$-embedding is measured via its distortion and its label-size/dimension.

We also study the analogous question for the prioritized versions of these two measures. Here, a priority order $\pi=(x_1,\dots,x_n)$ of the point set $X$ is given, and higher-priority points should have shorter labels. Formally, a distance labeling has prioritized label-size $\alpha(\cdot)$ if every $x_j$ has label size at most $\alpha(j)$. Similarly, an embedding $f: X \to \ell_\infty$ has prioritized dimension $\alpha(\cdot)$ if $f(x_j)$ is non-zero only in the first $\alpha(j)$ coordinates. In addition, we compare these prioritized measures
to their classical (worst-case) versions. 

We answer these questions in several scenarios, uncovering a surprisingly diverse range of behaviors. First, in some cases labelings and embeddings have very similar worst-case performance, but in other cases there is a huge disparity. However in the prioritized setting, we most often find a strict separation between the performance of labelings and embeddings. And finally, when comparing the classical and prioritized settings, we find that the worst-case bound for label size often ``translates'' to a prioritized one, 
but also find a surprising exception to this rule.
\end{abstract}
{\bf MSC classification codes:} 30L05, 46B85, 05C12, 05C78, 68R12


	\newpage
%
\section{Introduction}
It is often useful to succinctly represent the pairwise distances
in a metric space $(X,d)$ in a distributed manner.
A common model, called \emph{distance labeling},
assigns to each point $x\in X$ a label $l(x)$, 
such that some algorithm $\mathcal{A}$ (oblivious to $(X,d)$) can compute
the distance between any two points $x,y \in X$ given only their labels $l(x),l(y)$,
i.e., $\mathcal{A}\left(l(x),l(y)\right)=d(x,y)$. 
The goal is to construct a labeling whose label-size,
defined as $\max_{x\in X}|l(x)|$, is small. 
For general $n$-point metric spaces,
Gavoille, Peleg, P{\'{e}}rennes and Raz~\cite{GPPR04}
constructed a labeling scheme with 
label size of $O(n)$ words, and also proved this bound to be tight.%
\footnote{We measure size in words to avoid issues of bit representation.
  In the common scenario where distances are polynomially-bounded integers,
  every word has $O(\log n)$ bits, where $n=|X|$.
The bounds in \cite{GPPR04} are given in bits and are for unweighted graphs. 
Nevertheless, once we consider weighted graphs,
 $\Theta(n)$ words are sufficient and necessary for exact distance labeling, 
see \Cref{thm:GeneralGraphPriorirtyLabeling}.}

To obtain smaller label size, 
one often considers algorithms that approximate the distances. 
A distance labeling is said to have \emph{distortion} $t\ge 1$ if
\[
  \forall x,y\in X,
  \qquad
  d(x,y)\le \mathcal{A}\left(l(x),l(y)\right) \le t\cdot d(x,y).
\]
While the lower bound of \cite{GPPR04} holds even for distortion $t<3$, 
Thorup and Zwick~\cite{TZ05} constructed a labeling scheme with distortion 
$2t-1$ and label size $O(n^{1/t}\log n)$ for every integer $t\ge2$. 
These bounds are almost tight (assuming the Erd\H{o}s girth conjecture),
and demonstrate that for distortion $O(\log n)$, 
label size $O(\log n)$ is possible. \footnote{A much earlier technique to construct labeling scheme with distortion $O(\log n)$ is Bourgain's \cite{Bou85} embedding into $O(\log^2n)$-dimensional $\ell_2$, providing $O(\log^2n)$ label size.} 

From an algorithmic viewpoint, there is a significant advantage
to labels possessing additional structure, 
for example labels that are vectors in a normed space. 
This structure can lead to improved algorithms,
for example nearest neighbor search \cite{In98,BG19}.
A natural candidate for vector labels is the $\ell_\infty$ space, 
since every finite metric space embeds into it
isometrically (i.e., with no distortion). 
As such isometric embeddings require $\Omega(n)$ dimensions \cite{LLR95},
one may consider instead embeddings with small distortion. 
Formally, an embedding $f:X\rightarrow\ell_\infty$ is said to have
\emph{distortion} $t\ge1$ if 
\[
  \forall x,y\in X,
  \qquad
  d(x,y)\le \|f(x)-f(y)\|_\infty \le t\cdot d(x,y). 
\]
Matou{\v{s}}ek~\cite{Mat96} showed that for every integer $t\ge2$, 
every metric space embeds with distortion $2t-1$
into $\ell_\infty$ of dimension $O(n^{1/t}\cdot t\cdot \log n)$
(which again is almost tight assuming the Erd\H{o}s girth conjecture). 
For distortion $O(\log n)$,
Abraham \etal \cite{ABN11} later improved the dimension to $O(\log n)$.

In this paper, we take the perspective that $\ell_\infty$-embeddings
are a particular form of distance labelings,
and study the trade-offs these two models offer 
between distortion and dimension/label-size. 
While the inherent structure of $\ell_\infty$-embeddings makes them preferable, 
one may suspect that their additional structure precludes the tight trade-off
achieved using generic labelings.
Yet we have seen that for general metric spaces,
the performance of $\ell_\infty$-embeddings is essentially equivalent
to that of generic labelings.
This observation motivates us to consider more restricted input metrics, 
such as $\ell_p$ spaces, planar graph metrics, and trees. 
The central question we address is the following. 
\begin{question}\label{q1:LabVsEmb}
  In what settings are generic distance labelings
  more succinct than $\ell_\infty$-embeddings, 
  and how significant is the gap between them? 
\end{question}

\paragraph{Priorities.} 
Elkin, Filtser and Neiman~\cite{EFN18} introduced the problems of \emph{prioritized distortion} and \emph{prioritized dimension}; 
they posited that some points have higher importance or priority,
and it is desirable that these points achieve improved performance. 
Formally, given a priority ordering $\pi=\{x_1,\dots,x_n\}$ on the point set $X$, 
we say that embedding $f:X\rightarrow\ell_\infty$ possesses
\emph{prioritized contractive distortion} 
\footnote{In the original definition of prioritized distortion in \cite{EFN18},
the requirement of equation (\ref{eq:priorDistDef}) is replaced by the requirement 
$d(x_j,x_i)\le\|f(x_j)-f(x_i)\|_\infty\le \alpha(j)\cdot d(x_j,x_i)$.
We add the word \emph{contractive} to emphasize this difference. 
Prioritized contractive distortion is somewhat weaker in that
it does not imply scaling distortion (see \Cref{subsec:related}).
\label{foot:PrioritzedContractive}
}
$\alpha:\N\rightarrow\N$ (w.r.t.\ $\pi$) if 

\begin{eqnarray}
  \forall j<i,
  \qquad
	\frac{1}{\alpha(j)}\cdot d(x_j,x_i)\le \|f(x_j)-f(x_i)\|_\infty\le  d(x_j,x_i)~.\label{eq:priorDistDef}
\end{eqnarray}
Prioritized distortion is defined similarly for distance labeling. 
Furthermore, we say that a labeling scheme has prioritized label-size $\beta:\N\rightarrow\N$, if every $x_j$ has label length $|l(x_j)|\le \beta(j)$.
We say that embedding $f:X\rightarrow\ell_\infty$ has prioritized dimension $\beta$ if every $f(x_j)$ is non-zero only in the first $\beta(j)$ coordinates (i.e., $f_{i}(x_j)=0$ whenever $i>\beta(j)$). 
Here too $\ell_\infty$-embeddings are a more structured case of labelings,
and we again ask what are the possible trade-offs and how these two compare.
It is worth noting that the priority functions $\alpha,\beta$
are defined on all of $\N$ and apply when embedding every finite metric space;
in particular, they are not allowed to depend on $n=|X|$. 
Analogously to \Cref{q1:LabVsEmb},
we may also ask here about prioritized label size and dimension:

\begin{question}\label{q2:PriorVs}
	In what settings are distance labelings with prioritized label size
	more succinct than $\ell_\infty$ embeddings with prioritized dimension,
	and how significant is the gap between them? 
\end{question}

In many embedding results, the (worst-case) distortion
is a function of the size of the metric space $n=|X|$.
Elkin \etal \cite{EFN18} demonstrated a general phenomenon:
Often a worst-case distortion $\alpha(n)$ can be replaced
with a prioritized distortion $\tilde{O}(\alpha(j))$ using the same $\alpha$. \footnote{We use $\tilde{O}$ notation to suppress constants and logarithmic factors, that is $\tilde{O}(\alpha(j))=\alpha(j)\cdot\polylog(\alpha(j))$.}
For example, every finite metric space embeds into a distribution over trees
with prioritized expected distortion $O(\log j)$,
which extends the $O(\log n)$ distortion known from~\cite{FRT04}. 
Recently, Bartal \etal \cite{BFN19} showed that every finite metric space
embeds into $\ell_2$ with prioritized distortion $O(\log j)$,
which extends the $O(\log n)$ distortion known from~\cite{Bou85}.
In fact, we are not aware of any setting where it is impossible to generalize
a worst-case distortion guarantee to a prioritized guarantee. 
The final question we raise is the following. 
\begin{question}\label{q3:priorVersion}
  Does this analogy between worst-case and prioritized distortion 
  extend also to dimension and to label-size, 
  or perhaps their worst-case and prioritized versions exhibit a disparity? 
\end{question}

\subsection{Results: Old and New}
\begin{table}[p]
\centering
\begin{tabular}{|ll|l|lr|lr|}
	\hline \rowcolor[HTML]{FFCE93} 
	\multicolumn{7}{|c|}{\textbf{Worst-Case Label-Size/Dimension}}                                                                                              \\ \hline
	\rowcolor[HTML]{EFEFEF} 
	&& \textbf{Distortion} & \multicolumn{2}{|c|}{\textbf{Distance Labeling}}             & \multicolumn{2}{|c|}{\textbf{Embedding into $\ell_\infty$}}  \\ \hline
	1.&General Metric             & $<3$          & $\Theta(n)$ \tablefootnote{The upper bound is for distortion $1$ (i.e.\ isometric embedding).\label{foot:exact}} 
	&\cite{GPPR04}&$\Theta(n)$ &                              \cite{MatEmbedding13}\\ \hline
	\rowcolor[HTML]{EFEFEF} 
	2.&General Metric             & $O(\log n)$          & $O(\log n)$ 
	&\cite{TZ05}&$\Theta(\log n)$ &                              \cite{ABN11}\\ \hline		
	3.&$\ell_p$ for $p\in[1,2]$                  & $1+\eps$        & $O({\eps}^{-2}\log n)$                  &\theoremref{thm:L2toLinftyPriorUB}&$\Theta(n)$ \tablefootnote{Holds for $1+\eps<\sqrt{2}$ and $p\in[1,\infty]$.}  &                         \theoremref{thm:LBL2toLinfty}\\ 		\hline
	\rowcolor[HTML]{EFEFEF} 
	4.&Tree                & $1$        & $O(\log n )$                   & \cite{TZ01} &$\Theta(\log n )$   &                      \cite{LLR95}\\\hline
	5.&Planar                   & 1        & $\Theta(\sqrt{n})$  
	&\cite{GPPR04}&$\Theta(n)$&  \cite{LLR95} \\ \hline	
	\rowcolor[HTML]{EFEFEF} 
	6.&Treewidth $k$                   & 1        & $O(k\log n)$                 &\cite{GPPR04} &$\Theta(n)$ \tablefootnote{Holds for $k\ge2$.\label{foot:kge2}}&   \cite{LLR95} \\ \hline\hline	
	
	\hline \rowcolor[HTML]{FFCE93} 
	\multicolumn{7}{|c|}{\textbf{Prioritized Label-Size/Dimension}}                                                                                            \\ \hline
	\rowcolor[HTML]{EFEFEF} 
	&& \textbf{Distortion} & \multicolumn{2}{|c|}{\textbf{Distance Labeling}}             & \multicolumn{2}{|c|}{\textbf{Embedding into $\ell_\infty$}}  \\ \hline
	7.&General Metric             & $<\frac32$        &  $\Theta(j)$ $^{\ref{foot:exact}}$            &\theoremref{thm:GeneralGraphPriorirtyLabeling}& $\Theta(n)$ 
\tablefootnote{This excludes priority dimension for any function $\alpha:\N\rightarrow\N$ that is independent of $n=|X|$.\label{foot:noPrior}} &   
\theoremref{thm:GeneralGraphPriorityEmbedding}\\ \hline
	\rowcolor[HTML]{EFEFEF} 
	8.&General Metric & $O(\log j)$         & $O(\log j)$                           &\cite{EFN18} & $O(j)$ &                                        \corollaryref{cor:GeneralGraphPriorirtyDistortionEmbedding} \\ \hline
	9.&$\ell_p$ for $p\in[1,2]$                        & $1+\epsilon$        & $O(\epsilon^{-2}\log j)$               &\theoremref{thm:L2toLinftyPriorUB} & $j^{\Omega(\frac{1}{\eps})}$  &                \theoremref{thm:L2toLinftyPriorLB}\\ \hline
	\rowcolor[HTML]{EFEFEF} 
	10.&Tree                   & $1$        & $O(\log j)$                   &\cite{EFN18} & $\Theta(\log j)$ & \theoremref{thm:TreePriorLinfty}\\ \hline
	11.&Planar                   & $1$        & $\Theta(j)$                    &\theoremref{thm:PlanarPrioritzedLB} & $\Theta(n)$ $^{\ref{foot:noPrior}}$   &  \theoremref{thm:PlanarEmbeddingLB}\\ \hline	
	\rowcolor[HTML]{EFEFEF} 
	12.&Treewidth $k$                   & $1$        & $O(k\log j)$                   & \cite{EFN18} & $\Theta(n)$ $^{\ref{foot:kge2}}$ $^{\ref{foot:noPrior}}$ &       \theoremref{thm:PlanarEmbeddingLB}\\ \hline		
\end{tabular}
\caption{\small\it
  Summary of our findings. 
  \Cref{q1:LabVsEmb} is answered by comparing the last two columns of lines 1-6;
  in the very general and very restricted families (lines 1,2,4),
  labelings and embeddings perform similarly,
  while other families 
  (lines 3,5,6) exhibit a strict separation. 
  \Cref{q2:PriorVs} is answered by comparing the last two columns of lines 7-12;
  we see a strict separation between them 
  in all families other than trees (line 10). 
  \Cref{q3:priorVersion} is answered by comparing each line $i=1,\ldots,6$ with line $i+6$;
  for distance labeling, worst-case bound $\beta(n)$ translates to prioritized $O(\beta(j))$ except for planar graphs (lines 5,11),
  while for embeddings, dimension translates to its prioritized version only for trees (lines 4,10).  
  \label{tab:results}}	
\hrule 
\end{table}

\begin{table}[p]
	\centering
	\begin{tabular}{|ll|l|l|l|}
		\hline \rowcolor[HTML]{FFCE93} 
		\multicolumn{5}{|c|}{\textbf{Distance Labelings for General Metrics}}                                                                                              \\ \hline
		\rowcolor[HTML]{EFEFEF} 
		\multicolumn{2}{|l|}{\textbf{Prioritized Distortion}} &  \textbf{Prioritized  Label size}  & \textbf{Notes}&\textbf{Ref} \\ \hline
		1.&$2\cdot\left\lceil \nicefrac{k\log j}{\log n}\right\rceil -1$           &$O(n^{\nicefrac1k}\cdot\log j)$
		& $\forall k\in\N$&\cite{EFN18}\\ \hline
		\rowcolor[HTML]{EFEFEF} 
		2.&$2k-1$
		&  $O(j^{\nicefrac1k}\cdot\log j)$         & $\forall k\in\N$& \cite{EFN18}\\ \hline	\hline
		\rowcolor[HTML]{FFCE93} 
		\multicolumn{5}{|c|}{\textbf{Embeddings of General Metrics}}                                                                                              \\ \hline
		\rowcolor[HTML]{EFEFEF} 
		\multicolumn{2}{|l|}{\textbf{Prioritized Distortion}} &  \textbf{Prioritized  Dimension}  & \textbf{Notes}&\textbf{Ref} \\ \hline
		3.&$O(\log^{4+\eps}j)$           &$O(\log^{4}j)$
		& $\forall$ constant $\eps$&\cite{EFN18}\\ \hline
		\rowcolor[HTML]{EFEFEF} 
		4.&$2\cdot\left\lceil \nicefrac{k\log j}{\log n}\right\rceil -1$
		&  $O(k\cdot n^{\nicefrac{1}{k}}\cdot\log n)$         & $\forall k\in\N$& \cite{EN20}\\ \hline	
		&$2\cdot\lceil \log j\rceil -1$
		&  $O(\log^2 n)$         & & \cite{EN20}\\\hline
		\rowcolor[HTML]{EFEFEF} 5.&$2k\cdot\log\log j+1$
		& $O(k\cdot (j^{\nicefrac{2}{k}}+\log k)\cdot\log n)$                      & $\forall k\in\N$& \cite{EN20} \\\hline	\hline
		6.&$2\cdot\left\lceil \nicefrac{k\log j}{\log n}\right\rceil$
		&  $n^{\nicefrac{1}{k}}\cdot j$        &$\forall k\in\N$ & \corollaryref{cor:GeneralGraphPriorirtyDistortionEmbedding}\\ \hline			
		\rowcolor[HTML]{EFEFEF} 
		&$2\cdot\lceil\log j\rceil$
		&  $2j$        & & \corollaryref{cor:GeneralGraphPriorirtyDistortionEmbedding}\\ \hline

		7.&$2\cdot\lceil\log\log j\rceil$
		& $j^2$                     & & \corollaryref{cor:GeneralGraphPriorirtyDistortionEmbedding} \\ \hline
		
	\end{tabular}
	\caption{{\small\it
			$\ell_\infty$-embeddings and distance labelings of general metrics
			with different trade-offs between prioritized distortion and dimension/label size. 
			The labeling results are superior to their embedding counterparts. 
			Line 6 is obtained by plugging in $t=\frac{\log n}{k}$ in \Cref{cor:GeneralGraphPriorirtyDistortionEmbedding}.
			Comparing the result in line 4 to ours in line 6,
			in the most interesting regime of distortion $2\log j$, 
			we achieve a truly prioritized result (with dimension independent of $n$), 
			while \cite{EN20} avoids linear dependencies in the dimension.
			Our result in line 7 is strictly superior to that of line 5, which is not truly prioritized. 
			However, \cite{EN20} provides a much wider spectrum of possible trade-offs. \label{tab:priorDistortion}}}
	\hrule
\end{table}

Our main results and most relevant previous bounds
are discussed below and summarized in \Cref{tab:results}. 
Additional related work is described in \Cref{subsec:related}.

\paragraph{General Metrics.}

As discussed above, embeddings and labeling schemes for general metrics have essentially the same label size/dimension for all distortion parameters. 
For prioritized labelings and embeddings, the comparison is more complex. 
For exact labeling scheme, one can obtain label size $O(j)$
by simply storing in the label of the point $x_j$ its distances to $x_1,\dots,x_{j-1}$ (recall that we count words).
This is essentially optimal, even if we allow distortion up to $3$, see \Cref{thm:GeneralGraphPriorirtyLabeling}.
In contrast, for embeddings into $\ell_\infty$, we show in \Cref{thm:GeneralGraphPriorityEmbedding} that 
prioritized dimension is impossible for distortion less than $\frac32$.
Specifically, we provide an example where
the images of $x_1$ and $x_2$ must differ in at least $\Omega(n)$ coordinates
for arbitrarily large $n$.
This proves a strong separation between embeddings and labelings, 
and also demonstrates an embedding result that has no prioritized counterpart.

For prioritized distortion $O(\log j)$, Elkin \etal \cite{EFN18}
constructed a labeling with prioritized label size of $O(\log j)$. 
We construct in \Cref{thm:GeneralGraphPriorirtyDistortionEmbedding}
$\ell_\infty$-embeddings with different tradeoffs between 
the prioritized distortion $\alpha$ and dimension $\beta$. 
Two representative examples are prioritized distortion
$\alpha(j)=O(\log j)$ with prioritized dimension $\beta(j)=O(j)$, 
and $\alpha(j)=O(\log\log j)$ with $\beta(j)=O(j^2)$. 
This is significantly better than for the $O(1)$-distortion case,
yet considerably weaker than results on labeling.

Additional interesting results in this context were given in~\cite{EFN18}, 
showing that every metric space embeds into every $\ell_p$, $p\in[1,\infty]$,  with prioritized distortion $O(\log^{4+\eps}j)$
and prioritized dimension $O(\log^4j)$ (for every constant $\eps>0$).
Furthermore, independently and concurrently to our work,  Elkin and Neiman \cite{EN20} obtained two additional embeddings into $\ell_\infty$, for any integer parameter $k\ge 1$, there are embeddings with:
(1) prioritized distortion $2\left\lceil \nicefrac{k\log j}{\log n}\right\rceil -1$ and  dimension $O(k\cdot n^{\frac{1}{k}}\cdot\log n)$ (not prioritized); and
(2) prioritized distortion $2k\log\log j+1$ and prioritized
dimension $O(k\cdot j^{\frac{2}{k}}\cdot\log n)$ (note that the dimension bounds here also depend on $n=|X|$ 
and hence are not truly prioritized). 
See \Cref{tab:priorDistortion} for a comparison of these results with ours.

\paragraph{$\ell_p$ Spaces.}
The seminal Johnson-Lindenstrauss Lemma~\cite{JL84} states that
every $n$-point subset of $\ell_2$ embeds with distortion $1+\eps$
into $\ell_2^{O(\eps^{-2}\log n)}$
(where as usual $\ell_p^d$ denotes the $d$-dimensional $\ell_p$ space),
and this readily implies a labeling with distortion $1+\eps$ 
and label size $O(\eps^{-2}\log n)$.
Since every $\ell_p$, $p\in[1,2]$, embeds isometrically into squared-$L_2$
(equivalently, its snowflake embeds into $L_2$),
this implies a labeling with the same performance
for $\ell_p$ as well, see \Cref{thm:L2toLinftyPriorUB}. 
Furthermore, we show in \Cref{thm:L2toLinftyPriorUB} (using \cite{NN19}) that
this labeling can be prioritized to achieve
distortion $1+\eps$ with label size $O(\eps^{-2}\log j)$.

For $\ell_\infty$- embeddings, the performance is significantly worse. 
We show in \Cref{thm:LBL2toLinfty} that 
for certain $n$-point subsets of $\ell_p$, for any $p\in[1,\infty]$,
embedding into $\ell_\infty$ with distortion less than $\sqrt{2}$
requires $\Omega(n)$ coordinates 
(recall that $O(n)$ coordinates are sufficient to isometrically embed
every $n$-point metric into $\ell_\infty$.)
For prioritized embeddings into $\ell_\infty$ with distortion $1+\eps$, 
we prove a lower bound of $j^{\Omega(\frac1\eps)}$ on the prioritized dimension, see \Cref{thm:L2toLinftyPriorLB}.

\paragraph{Tree Metrics.} 
Trees are a success story, where both labelings and embeddings have the same performance. 
Here we study metric spaces that induced by the shortest path metric of weighted trees.
In their seminal paper on metric embeddings,
Linial, London and Rabinovich~\cite{LLR95} proved that every $n$-node tree
embeds isometrically into $\ell_\infty^{O(\log n)}$. 
In the context of routing, Thorup and Zwick~\cite{TZ01} constructed
an exact distance labeling with label size $O(\log n)$ 
(where routing decisions can be done in constant time),
and Elkin \etal \cite{EFN18} modified this
to achieve prioritized label size $O(\log j)$. 
Our contribution is a prioritized version of~\cite{LLR95},
i.e., an isometric embedding of a tree metric into $\ell_\infty$
with prioritized dimension $O(\log j)$, see \Cref{thm:TreePriorLinfty}. 
We note that an equivalent result was proved independently and concurrently
by Elkin and Neiman~\cite{EN20}.

\paragraph{Planar Graphs and Restricted Topologies.}
Here we study metric spaces that induced by the shortest path metric of weighted graphs with restricted topologies.
We first consider exact distance labeling and isometric embeddings.
Gavoille \etal \cite{GPPR04} showed that planar graphs
admit exact labeling with label size $O(\sqrt{n})$,
and proved a matching lower bound.%
\footnote{This lower bound, as well as all other lower bounds
  from~\cite{GPPR04}, count bits rather than words. 
}
They further showed that treewidth-$k$ graphs admit exact labeling
with label size $O(k\log n)$. 
Linial \etal \cite{LLR95} proved that an isometric embedding
of the $n$-cycle graph into $\ell_\infty$, and in fact into any normed space,
requires $\Omega(n)$ coordinates.%
\footnote{Their proof is much more general than what is required for $\ell_\infty$. 
  For a simpler proof for the special case studied here, see \Cref{thm:PlanarEmbeddingLB}.
}
Notice that the cycle graph is both planar and has treewidth $2$;
hence, there is a strict separation between distance labeling
and $\ell_\infty$-embedding. 

For exact prioritized distance labeling, we prove that planar graphs
require prioritized label size $\Omega(j)$ (based on \cite{GPPR04}),
see \Cref{thm:PlanarPrioritzedLB}. 
This bound is tight, as prioritized label size $O(j)$ is possible
already for general graphs (\Cref{thm:GeneralGraphPriorirtyLabeling}). 
We conclude that priorities make exact distance labelings
much harder for planar graphs.%
\footnote{Interestingly, for unweighted planar graphs,
  Gavoille \etal \cite{GPPR04} prove only a lower bound
  of $\Omega(n^{\frac13})$ on the label size,
  and closing the gap to the upper bound $O(\sqrt{n})$
  remains an important open question.
} 
This lower bound for exact prioritized labeling holds for unweighted graphs
as well, hence this type of labeling is now well understood.
For treewidth-$k$ graphs, Elkin \etal \cite{EFN18}
constructed exact labeling with prioritized label size $O(k\log j)$.
For isometric embeddings into $\ell_\infty$,
we show in \Cref{thm:PlanarEmbeddingLB}
that no prioritized dimension is possible for the cycle graph, 
which provides a lower bound for both planar and treewidth-$2$ graphs. 
This implies a dramatic separation for these families.

Additional results on labelings with $1+\eps$ distortion, and embeddings with constant distortion are described in \Cref{subsec:related}.

\subsection{Related Work}\label{subsec:related} 
For distortion $1+\eps$ in planar graphs, Thorup~\cite{Tho04}
and Klein~\cite{Kle02} constructed distance labels of size $O(\log n/\eps)$. 
Abraham and Gavoille~\cite{AG06} generalized this result to $K_r$-minor-free graphs, achieving label size $O(g(r)\log n/\eps)$. 
\footnote{The function $g(r)$ depends only on $r$ and is taken from the  structure theorem of Robertson and Seymour~\cite{RS03}\label{foot:RS03}.}
No low-dimension embedding into $\ell_\infty$ with distortion $1+\eps$
is known for planar graphs or even treewidth-$2$ graphs. 
If one allows larger distortion,
Krauthgamer \etal \cite{KLMN05} proved that planar graphs
embed with distortion $O(1)$ into $\ell_\infty^{O(\log n)}$,
or more generally that $K_r$-minor-free graphs
embed with distortion $O(r^2)$ into $\ell_\infty^{O(3^r\cdot\log n)}$. 
Abraham \etal \cite{AFGN18} showed that $K_r$-minor-free graphs embed
with distortion $O(1)$ into $\ell_\infty^{O(g(r)\log^2 n)}$.
Turning to priorities, Elkin \etal \cite{EFN18} constructed
prioritized versions of distance labeling with distortion $1+\eps$. 
Specifically, for planar and $K_r$-minor-free graphs they achieve label sizes
of $O(\log j/\eps)$ and $O(g(r)\log j/\eps)$, respectively. 
No prioritized embeddings are known, nor lower bounds thereof.

Elkin \etal \cite{EFN17} studied the problem of {\em terminal} distortion, 
where there is specified a subset 
$K\subset X$ of terminal points, and the goal is to embed the entire space $(X,d)$ while preserving pairwise distances among 
$K\times X$. 
For additional embeddings with terminal distortion see \cite{EN18,BFN19}.
Embeddings with terminal distortion can be used used to construct embeddings with prioritized distortion. 
We utilize this approach in~\Cref{thm:L2toLinftyPriorUB,thm:TreePriorLinfty}.

Abraham \etal \cite{ABN11} studied {\em scaling} distortion, which provides improved distortion for $1-\eps$ fractions of the pairs, 
simultaneously for all $\eps\in(0,1)$, as a function of $\eps$. A stronger version called {\em coarse scaling} distortion has 
improved distortion guarantees for the farthest pairs. Bartal \etal \cite{BFN19} showed that scaling distortion and prioritized 
distortion (in the sense of \cite{EFN18}) are essentially equivalent, but this is not known to hold for the prioritized contractive 
distortion we use in the current paper (see footnote (\ref{foot:PrioritzedContractive})).

\subsection{Preliminaries}
The $\ell_p$-norm of a vector $x=(x_1,\dots,x_d)\in \R^d$ is $\Vert
x \Vert_{p} : =(\sum_{i=1}^{d}|x_{i}|^{p})^{1/{p}}$, where $\Vert
x\Vert_{\infty} :=\max_{i}|x_{i}|$.
An embedding $f$ between two metric spaces $(X,d_X)$ and $(Y,d_Y)$ has distortion $c\cdot t$ if for every $x,y\in X$, $\frac1c\cdot d_X(x,y)\le d_Y(f(x),f(y))\le t\cdot d_X(x,y)$. $t$ (resp. $c$) is the \emph{expansion} (resp. \emph{contraction}) of $f$. If the expansion is $1$, we say that $f$ is \emph{Lipschitz}, while if $c=1$ we say that the embedding is \emph{non-contractive}. Embedding with distortion $1$ (where $c=t=1$) is called \emph{isometric}.

Embedding $f:X\rightarrow \ell_\infty^d$ can be viewed as a collection of embeddings $\{f_i\}_{i=1}^d$ into the line $\R$. By 
scaling we can assume that the embedding is non-contractive. That is, if $f$ has distortion $t$ then for every $x,y\in X$ and $i$, $|f_i(x)-f_i(y)|\le t\cdot d_X(x,y)$ and there is some index $i_{x,y}$ such that $d_X(x,y)\ge 
|f_{i_{x,y}}(x)-f_{i_{x,y}}(y)|$. We say that the pair $x,y$ is \emph{satisfied} by the coordinate $i_{x,y}$.

We consider connected undirected graphs $G=(V,E)$ with edge weights $w: E \to \R_{> 0}$. Let $d_{G}$ denote the shortest path metric in $G$. For a vertex $x\in V$ and a set 
$A\subseteq V$, let $d_{G}(x,A):=\min_{a\in A}d(x,a)$, where $d_{G}(x,\emptyset):= \infty$. We often abuse notation and write the graph $G$ instead of its vertex set $V$.

We always measure the size of a label by the number of words needed to store it (where each word contains $O(\log n)$ bits). 
For ease of presentation, we will ignore issues of representation and bit counting. 
In particular, we will assume that every pairwise distance can be represented in a single word. 
We note however that the lower bounds of \cite{GPPR04} are given in bits, and therefore our
\Cref{thm:PlanarPrioritzedLB} is as well.

All logarithms are in base $2$. Given a set $A$, ${A\choose 2}=\left\{\{x,y\}\mid x,y\in A,x\ne y\right\}$ denotes all the subsets of size $2$.
The notation $x=(1\pm\eps)\cdot y$ means $(1-\eps)y\le x\le (1+\eps)y$.

\section{General Graphs}\label{sec:generalGraphs}
In this section we discuss our result on succinct representations of general metric spaces.
We start with the regime of small distortion. Recall that there exist both exact distance labelings with $O(n)$ label size \cite{GPPR04} 
as well as isometric embeddings into 
$\ell_\infty^n$ \cite{MatEmbedding13}, and both results are essentially tight (even if one allows distortion $<3$).
In the following theorem we provide lower and upper bounds for exact distance labelings with prioritized label size. 
\begin{theorem}\label{thm:GeneralGraphPriorirtyLabeling}
Given an $n$-point metric space $(X,d)$ and priority ordering $X=\{x_1,\dots,x_n\}$, 
there is an exact labeling scheme with prioritized label size $j$. This is asymptotically tight, that is every exact labeling scheme must have prioritized label size $\Omega(j)$.
Furthermore, for $t<3$, every labeling scheme with distortion $t$ must have prioritized label size $\tilde{\Omega}(j)$.
\end{theorem}
\begin{proof}
	We begin by constructing the labeling scheme. 
	The label of $x_j$ simply consists of the index $j$ and $d(x_1,x_j),d(x_2,x_j),\dots,d(x_{j-1},x_{j})$. 
	The size bound and algorithm for answering queries are straightforward.
	If one allows distortion $t<3$, \cite{GPPR04} proved that every labeling scheme with distortion $t$ must have label size of 
	$\Omega(n)$ bits, or $\tilde{\Omega}(n)$ words. As some vertex must have a label of size $\tilde{\Omega}(n)$, the prioritized 
	lower bound  $\tilde{\Omega}(j)$ follows. 
	
	Finally, we prove the $\Omega(j)$ lower bound for exact distance labeling. We begin by arguing that some label must be of length 
	$\Omega(n)$ (in words), and then the $\Omega(j)$ lower bound for prioritized label size follows.
	The proof follows the steps of \cite{GPPR04}. 
	Consider a complete graph with ${n\choose 2}$ edges all having integer weights in $\{n+1,n+2,\dots,2n\}$. 
	Note that there are $n^{{n\choose 2}}$ such graphs, where each choice of weights defines a different shortest path metric. 
	Given an exact labeling scheme, the labels $l(x_1),\dots,l(x_n)$ precisely encode the graph. 
	Following arguments from \cite{GPPR04}, the sum of lengths of the labels must be at least logarithmic in the number of different graphs. 
	Thus
	$$\max_i|l(x_i)|\ge\frac1n\cdot\log n^{{n\choose 2}} =\Omega(n\log n)~.$$
	We conclude that some label length must be of $\Omega(n\log n)$ bits, or $\Omega(n)$ words.
\end{proof}
	
While under the standard worst-case model distance labelings and embeddings into $\ell_\infty$ behave identically, 
we show that the prioritized versions are very different. In the following 
theorem we show that no prioritized dimension is possible, even if one allows distortion $<\frac32$ (note that for  much larger distortions, prioritized dimension is possible. See \cite{EFN18} and \Cref{cor:GeneralGraphPriorirtyDistortionEmbedding}). 
	
	\begin{theorem}\label{thm:GeneralGraphPriorityEmbedding}
	There is no function $\alpha:\N\rightarrow\N$ such that every metric space can be
	embedded into $\ell_{\infty}$ with prioritized dimension $\alpha$
	and distortion $t<\frac{3}{2}$ (for any fixed $t$).
	\end{theorem}
	
	\begin{proof}
	Consider the family $\mathcal{G}$ of unweighted bipartite graphs
	$G=\left(V=L\cup R,E\right)$ where $|L|=|R|=n$, for large enough $n$. 
	We first argue that there is a graph $G\in\mathcal{G}$ with the following properties:
	\begin{enumerate}
	\item[(1)] For every $u,v\in R$ or $u,v\in L$, we have $d_{G}(u,v)=2$.
	\item[(2)] Every embedding $f:G\rightarrow\ell_{\infty}$ with distortion $2t$
	requires $\Omega(n)$ coordinates.
	\end{enumerate}
	The existence of $G$ follows by a counting argument similar to \cite{MatEmbedding13}. Note that $|\mathcal{G}|=2^{n^2}$.
	Denote by $\mathcal{G}'\subseteq \mathcal{G}$ the graphs in  $\mathcal{G}$ fulfilling property (1).
	Our first step is to lower bound $|\mathcal{G}'|$. Sample uniformly a graph $G\in\mathcal{G}$. 
	For $u,v\in R$ (resp. $u,v\in L$) let $I_{u,v}$ be an indicator for the event $d_G(u,v)\ne 2$. 
	$I_{u,v}$ occurs if and only if $u$ and $v$ do not have a common neighbor in $L$ (resp. $R$). Then $\Pr[I_{u,v}]=(\frac34)^n$. 
	By a union bound, the probability that property (1) does not hold is at most $2\cdot{n\choose 2}\cdot (\frac34)^n$. We conclude 
that  $|\mathcal{G}'|\ge 2^{n^2}\cdot \left(1-2\cdot{n\choose 2}\cdot (\frac34)^n\right)\ge\frac12\cdot2^{n^2}$.
	Matou{\v{s}}ek \cite{MatEmbedding13} (Proposition 3.3.1) implicitly proved that for any subset $\mathcal{G}'$ of $\mathcal{G}$, if all of $\mathcal{G}'$ embeds into $\ell_\infty^d$ with distortion $2t<3$, then
	$$c^{d\cdot n}\ge |\mathcal{G}'|~,$$
where $c>1$ is a constant depending on $3-2t$. Thus $d=\Omega(n)$. We conclude that there is a graph $G\in\mathcal{G}$ fulfilling both properties $(1),(2)$.

Consider such a graph $G=\left(V=L\cup R,E\right)$.
Note that property (1) implies that there are no isolated vertices, and moreover for every $u\in R$, $v\in L$,
	$d_{G}(u,v)\in\left\{ 1,3\right\}$. 
Let $G'$ be the graph
$G$ along with two new vertices $l,r$ where $l$ (resp.\ $r$)
	has edges to all vertices in $R$ (resp.\ $L$). Note that for every
	$u,v\in V$, $d_{G}(u,v)=d_{G'}(u,v)$. Set $L'=L\cup\{l\}$ and $R'=R\cup\{r\}$. 
	\begin{claim}\label{clm:GeneralNoPriority}
	Every embedding $f:G'\rightarrow\ell_{\infty}$ with distortion $t < \frac{3}{2}$
	has $\Omega(n)$ coordinates $i$ for which $f_{i}(l)\ne f_{i}(r)$. 
	\end{claim}
	\begin{proof}
	We assume that the embedding $f$ has expansion at most $t$, and for every pair of vertices there is a coordinate where the pair is satisfied (i.e.\ not contracted).
	\sloppy Set $\mathcal{A}_{i}=\left\{ \left\{ u,v\right\} \in{L'\cup R' \choose 2}\mid d_{G'}(u,v)=i\right\}$ to be all the vertex pairs at distance exactly $i$.
	Note that ${L'\cup R' \choose 2}=\mathcal{A}_{1}\cup\mathcal{A}_{2}\cup\mathcal{A}_{3}$.
	To satisfy all the pairs in ${L'\cup R' \choose 2}$, $\Omega(n)$ coordinates
	are required (this is property (2)). We will show that we can satisfy
	all the pairs in $\mathcal{A}_{1}\cup\mathcal{A}_{2}$ using $O\left(\log n\right)$ coordinates only. 
Thus satisfying all the pairs in $\mathcal{A}_{3}$ requires $\Omega(n)$ coordinates.

The clique $K_n$ can be embedded isometrically into $\ell_\infty^{\lceil\log n\rceil}$ \cite{LLR95}. Such an embedding can be 
constructed by simply mapping $K_n$ to different combinations of $\{0,1\}^{\lceil\log n\rceil}$. 
As $1$ is the minimal distance, we can just embed all the $2n+2$ vertices as a clique using $O(\log(n))$ coordinates. 	
By doing so, all the pairs in $\mathcal{A}_{1}$ will be satisfied. 
$\mathcal{A}_{2}$ equals ${L' \choose 2}\cup{R' \choose 2}$.
Note that the metric induced on ${L' \choose 2}$ is just a clique with edges of length 2. 
Thus we can embed all of $L'$ to the vectors $\{\pm1\}^{O(\log n)}$. Additionally send all of $R'$ 
	to $\vec{0}$. Note that by doing so we satisfied all the pairs in ${L' \choose 2}$ while incurring no expansion. Similarly we can satisfy all the pairs in ${R' \choose 2}$ using 
$O(\log n)$ additional coordinates.

	Next consider an arbitrary embedding $f:G'\rightarrow\ell_{\infty}$ with distortion $t < \frac{3}{2}$.
	We argue that in a coordinate $f_i:G'\rightarrow \R$ where $f_i(l)=f_i(r)$,
	no pair of $\mathcal{A}_{3}$ is satisfied. Indeed, every vertex
	$v\in L'\cup R'$ is at distance $1$ from either $l$ or $r$. As we have expansion at most $t$,  in a coordinate
	$i$ where $f_i(l)=f_i(r)$ the maximal distance
	between a pair of vertices $v,u$ is $2\cdot t$. In particular, for every
	$\left\{ v,u\right\} \in\mathcal{A}_{3}$, $\left|f_i(x)-f_i(y)\right|\le2\cdot t<3$. Thus no pair $\{v,u\}\in\mathcal{A}_3$ is satisfied.
		
	As there must be $\Omega(n)$ coordinates where some pair from $\mathcal{A}_{3}$
	is satisfied, necessarily there are $\Omega(n)$ coordinates where
	$f_i(l)\ne f_i(r)$.
	\end{proof}

	We conclude that there are $\Omega(n)$ coordinates where at least one of $l,r$ is not mapped to $0$. 
	Set $\pi$ to be any priority ordering wherein $l$ and $r$ have priorities $1$ and $2$ respectively.
	For every priority function $\alpha:\N\rightarrow\N$, by taking $n\gg\alpha(2),\alpha(1)$, there is no embedding with prioritized dimension $\alpha$ 
	with respect to $\pi$. The theorem follows.
	\end{proof}

	Considering that  for distortion less than $\frac32$ no prioritized dimension is possible, 
it is natural to ask for what distortions are prioritized embeddings possible.
	Some previous results of this nature are described in the introduction \cite{EFN18,EN20}. 
As exact distance labeling is possible using $O(j)$ labels, it is also
natural to ask what distortion can be obtained with prioritized dimension $O(j)$.
	The following is a meta theorem constructing various trade-offs.
	We present some specific implications in \Cref{cor:GeneralGraphPriorirtyDistortionEmbedding}. A comparison between our results and other results appears in \Cref{tab:priorDistortion}.
	
	Consider a monotone function $\beta:\N\rightarrow\N$. 
	For $j\in\N$, let $\chi_\beta(j)$ be the minimal $i$ such that $\beta(\chi_\beta(j))\ge j$.
	\begin{theorem}\label{thm:GeneralGraphPriorirtyDistortionEmbedding}
	Given a metric space $(X,d)$ with priority ordering $X=\{x_1,\dots,x_n\}$ and a function $\beta:\N\rightarrow\N$, there is an embedding $f:X\rightarrow\ell_\infty$ with  prioritized dimension $\beta(\chi_\beta(j))$ and contractive prioritized distortion $2\cdot\chi_\beta(j)$. 
	\end{theorem}
	
Before presenting the proof of \Cref{thm:GeneralGraphPriorirtyDistortionEmbedding},
we provide some of the intuition behind it. 
Recall that the Fr\'echet embedding \cite{MatEmbedding13} (also called Kuratowski embedding)
is an embedding into $\ell_\infty^n$, 
where the $j$'th coordinate for a point $x$ is simply $x$'s distance to $x_j$. 
While this is an isometric embedding, 
every point is non-zero in $n-1$ coordinates. 
In order to obtain prioritized dimension, 
we will set the $j$ coordinate of a point $x$ to be its distance to the set
that contains $x_j$ together with all points $x_{j'}$ for sufficiently small $j'$ 
(where the value of $j'$ is determined by the function $\beta$).
This ``padding'' will ensure prioritized dimension, 
but also induce larger distortion as a function of $\beta$.

\begin{proof}[Proof of \Cref{thm:GeneralGraphPriorirtyDistortionEmbedding}]
We suggest that while inspecting the proof, it may be helpful for the reader to focus on the setting $\beta(i)=2^i$, wherein $\chi_\beta(j)=\lceil\log j\rceil$.
	Set $S_0=\emptyset$ and $S_i=\{x_j\mid j\le \beta(i)\}$. We define embedding $f$ by setting its $j$'th coordinate to be
	\[f_j(x):=d(x,S_{\chi_\beta(j)-1}\cup\{x_j\})~.\]
	Note that for every $j'$ such that $\chi_\beta(j')>\chi_\beta(j)$, $f_{j'}(x_j)=0$. 
	Note also that there may be many points $x_{j'}$ with $j'<j$ and yet $f_j(x_{j'})\ne 0$.
	Thus $x_j$ is non-zero only in the first $\beta(\chi_\beta(j))$ coordinates as required.
	
	Next we argue the prioritized distortion. It is clear that $f$ is Lipschitz. Consider a pair of vertices $x_j,y$.
	Set $\Delta=d(x_j,y)$, and $\alpha_i=d\left(\{x_j,y\},S_i\right)$.
	Then $\infty=\alpha_0>\alpha_1\ge\alpha_2\ge\dots\ge\alpha_{\chi_\beta(j)}=0$. 
We argue that there must be some index $i$ such that
$\alpha_{i+1}\le\min\{\alpha_{i},\frac\Delta2\}-\frac{\Delta}{2\chi_\beta(j)}$.
Suppose for contradiction otherwise (i.e. no such index exist). 
We argue by induction on $q\in[1,\chi_\beta(j)]$ that $\alpha_{\chi_{\beta}(j)-q}<q\cdot\frac{\Delta}{2\chi_{\beta}(j)}$. 
For the base case note that $0=\alpha_{\chi_{\beta}(j)}>\min\{\alpha_{\chi_{\beta}(j)-1},\frac{\Delta}{2}\}-\frac{\Delta}{2\chi_{\beta}(j)}$,
implying $\alpha_{\chi_{\beta}(j)-1}<\frac{\Delta}{2\chi_{\beta}(j)}$. 
For general $q$, using the induction hypothesis $q\cdot\frac{\Delta}{2\chi_{\beta}(j)}>\alpha_{\chi_{\beta}(j)-q}>\min\{\alpha_{\chi_{\beta}(j)-q-1},\frac{\Delta}{2}\}-\frac{\Delta}{2\chi_{\beta}(j)}$,
implying $\min\{\alpha_{\chi_{\beta}(j)-q-1},\frac{\Delta}{2}\}<(q+1)\cdot\frac{\Delta}{2\chi_{\beta}(j)}$
and thus $\alpha_{\chi_{\beta}(j)-q-1}<(q+1)\cdot\frac{\Delta}{2\chi_{\beta}(j)}$. 
Overall we conclude that 
$\alpha_{0}=\alpha_{\chi_{\beta}(j)-\chi_{\beta}(j)}<\chi_{\beta}(j)\cdot\frac{\Delta}{2\chi_{\beta}(j)}=\frac{\Delta}{2}$, a contradiction as $\alpha_{0}=\infty$.

Choose $z\in S_{i+1}$ such that $d(\{x_j,y\},z)=\alpha_{i+1}$, and suppose that $z=x_q$. 
Assume without loss of generality that
$d(x_j,z)=d(\{x_j,y\},z)=\alpha_{i+1}$,
and so 
$d(y,z)\ge d(x_j,y)-d(x_j,z) \ge \Delta - \frac{\Delta}{2} + \frac{\Delta}{2\chi_\beta(j)} >\frac\Delta2$. 
	It holds that $d\left(y,S_{i}\cup\{z\}\right)=\min\left\{ d(y,S_{i}),d(y,z)\right\} \ge\min\left\{ \alpha_{i},\frac{\Delta}{2}\right\}$.
	Thus	
\begin{align*}
	\left\Vert f(y)-f(x_{j})\right\Vert _{\infty} & \ge\left|f_{q}(y)_{\infty}-f_{q}(x_{j})\right|\\
	& =\left|d\left(y,S_{i}\cup\left\{ z\right\} \right)-d(x_{j},S_{i}\cup\left\{ z\right\} )\right|\\
	& \ge\left|\min\left\{ \alpha_{i},\frac{\Delta}{2}\right\} -\alpha_{i+1}\right|\ge\frac{\Delta}{2\chi_{\beta}(j)}~.
\end{align*}	
	Prioritized distortion $2\cdot\chi_\beta(j)$ follows. 
	\end{proof}
	\begin{corollary}\label{cor:GeneralGraphPriorirtyDistortionEmbedding}
	Given a metric space $(X,d)$ with priority ordering $X=\{x_1,\dots,x_n\}$, 
	\begin{enumerate}
	\item For every $t\in\N$, there is an embedding $f:X\rightarrow\ell_\infty$ with prioritized distortion $2\cdot \lceil\frac{\log j}{t}\rceil$ and  prioritized dimension $2^{t}\cdot j$.			
	\item There is an embedding $f:X\rightarrow\ell_\infty$ with prioritized distortion $2\cdot\lceil\log\log j\rceil$ and  prioritized dimension $j^2$.
\end{enumerate}
\end{corollary}
\begin{proof}

The first case follow by choosing the function $\beta(i)=2^{t\cdot i}$. Here $\chi_\beta(j)=\lceil\log_{2^t} 
j\rceil=\lceil\frac{\log j}{t}\rceil$, and thus the prioritized distortion is $2\cdot \lceil\frac{\log j}{t}\rceil$ while the 
prioritized dimension is $\beta(\chi_{\beta}(j))=2^{t\cdot\lceil\frac{\log j}{t}\rceil}<2^{t+\log j}=2^{t}\cdot j$.

For the second case choose $\beta(i)=2^{2^i}$. Here $\chi_\beta(j)=\lceil\log\log j\rceil$, and thus the prioritized distortion is 
$2\cdot \lceil\log\log j\rceil$ and the prioritized dimension is $\beta(\chi_{\beta}(j))=2^{2^{\lceil\log\log j\rceil}}<2^{2\cdot2^{\log\log j}}=j^{2}$.
\end{proof}
Note that the first case implies prioritized distortion $2\cdot\lceil\log j\rceil$ and prioritized dimension $2j$.

\section{$\ell_p$ Spaces}\label{sec:lp}
In this section we consider representations of $\ell_p$ spaces. 
As these spaces are somewhat restricted, we focus on the $1+\eps$ distortion regime.
We begin with the upper bound for distance labeling. 
\begin{theorem}\label{thm:L2toLinftyPriorUB}
	For every $\eps>0$, $p\in[1,2]$ and $n$ points in $\ell_p$, there is a $(1+\eps)$-labeling scheme with label size $O(\eps^{-2}\log n)$. Furthermore, given a priority ordering $\pi$, there is a $(1+\eps)$-labeling scheme with prioritized label size $O(\eps^{-2}\log j)$.
\end{theorem}
\begin{proof}[Proof of \Cref{thm:L2toLinftyPriorUB}]
	We begin by constructing a labeling scheme for a set $X$ on $n$ points in $\ell_2$. Then we will generalize the result to $\ell_p$ for $p\in [1,2]$.
	
	As a consequence of the Johnson Lindenstrauss lemma \cite{JL84}, there is an embedding $f:X\rightarrow\ell_2^{O(\eps^{-2}\log n)}$ with $1+\eps$ distortion. By simply storing $f(x)$ as the label of $x\in X$, we obtain a $1+\eps$ labeling scheme with $O(\eps^{-2}\log n)$ label size.
	
	Next we consider a set $X$ with priority ordering $\pi=\{x_1,x_2,\dots,x_n\}$. 
	Narayanan and Nelson \cite{NN19} constructed a terminal version of the JL transform: 
	Specifically, given a 
	set $K$ of $k$ points in $\ell_2$ there is an embedding $f$ of the entire $\ell_2$ space into $\ell_2^{O(\eps^{-2}\log k)}$  such that for every $x\in K$ and $y\in \ell_2$, $\|f(x)-f(y)\|_2=(1\pm\eps)\|x-y\|_2$.
	
	For $i=0,1,\dots \lceil\log\log n\rceil$, set $S_i=\{x_j\mid j\le 
2^{2^{i}} \}$. Let $f_i:X\rightarrow\ell_2^{O(\log|S_i|)}$ be a terminal JL transform w.r.t.\ $S_i$. The label of $x_j$ will consist of $f_0(x_j),f_1(x_j),\dots,f_{\lceil\log\log j\rceil}(x_j)$.
	Given a query on $x_j,x_{j'}$, where $j<j'$, our answer will be $\|f_{\lceil\log\log j\rceil}(x_j)-f_{\lceil\log\log j\rceil}(x_{j'})\|_2$.
	The distortion follows as $x_j\in S_{\lceil\log\log j\rceil}$ (hence \cite{NN19} guarantees $\|f_{\lceil\log\log j\rceil}(x_j)-f_{\lceil\log\log j\rceil}(x_{j'})\|_2=(1\pm\eps)\|x_{j}-x_{j'}\|_{2}$).
	The length of the label of $x_j$ is bounded by
	\[
	\sum_{i=0}^{\lceil\log\log j\rceil}O(\eps^{-2}\log|S_i|)=O(\eps^{-2})\cdot \sum_{i=0}^{\lceil\log\log j\rceil}2^i=O(\eps^{-2})\cdot 2^{\lceil\log\log j\rceil+1}=O(\eps^{-2}\cdot\log j)~,
	\]
	words, as required.
	
	To generalize the labeling schemes to $\ell_p$ for $p\in[1,2]$, we note that 
	every $p\in[1,2]$, $\ell_p$ embeds isometrically into squared-$L_2$,
	or equivalently, the snowflake of $\ell_p$ embeds into $L_2$ (see e.g. \cite{DL97}). Specifically, for a set $X\subseteq\ell_p$, there is a function $f_X:X\rightarrow\ell_2$, such that for every $x,y\in X$, $\|x-y\|_p=\|f(x)-f(y)\|_2^2$.
	Then a labeling scheme for $\ell_2$ implies the same performance for $\ell_p$ as well,
	the only change being that the computed distances must be squared.
\end{proof}

Next we turn our attention to lower bounds. Every $n$-point set in $\ell_2$ embeds isometrically into any other $\ell_p$ space, for $p\in[1,\infty]$ (see e.g. \cite{MatEmbedding13}).
This implies that any lower bound that we prove for $\ell_2$ will holds as well for any other $\ell_p$ space (as the hard example will reside in $\ell_p$ as well).
\begin{theorem}\label{thm:LBL2toLinfty}
	For every $p\in[1,\infty]$ and $n\in\N$, there is a set $A$ of $2n$ points in $\ell_p$, such that every embedding of $A$ into $\ell_\infty$ with distortion smaller than $2^{\max\{\frac12,1-\frac1p\}}$ has dimension at least $n$.
\end{theorem}
\begin{proof}[Proof of \Cref{thm:LBL2toLinfty}]
	Set $A=\{e_1,-e_1,e_2,-e_2,\dots,e_n,-e_n\}$, the standard basis and its \emph{antipodal} points (here $\{e_i,-e_i\}$ is an 
\emph{antipodal pair}). 
	Fix $p$, and we will prove that every embedding of $A\subseteq\ell_p$ with distortion smaller than $2^{1-\frac1p}$ into $\ell_\infty$ requires  
	at least $n$ coordinates. As mentioned above, the lower bound for $p=2$ holds for all $\ell_p$ as well; thus the theorem will follow.

	\sloppy We argue that each coordinate can satisfy at most a single antipodal pair. As there are $n$ such pairs, the lower bound follows. 	
	Consider a single coordinate $f:A\rightarrow\R$. Assume by way of contradiction that there are $e_i,-e_i,e_j,-e_j\in A$ ($i\ne j$) such that $2\le 
	\left|f(e_i)-f(-e_i)\right|,\left|f(e_j)-f(-e_j)\right|$.
	As $f(e_i),f(-e_i),f(e_j),f(-e_j)\in \R$, by case analysis there must be a pair consisting of one point from $\{f(e_i),f(-e_i)\}$, and one point from $\{f(e_j),f(-e_j)\}$ at distance at least
	$\min\left\{\left|f(e_i)-f(-e_i)\right|,\left|f(e_j)-f(-e_j)\right|\right\}\ge2$. But the actual distance between this pair is only $2^{\frac1p}$. Thus $f$ has distortion $\frac{2}{2^{\nicefrac1p}}=2^{1-\frac1p}$, a contradiction.	
\end{proof}	

Note that \Cref{thm:LBL2toLinfty} implies a lower bound of $\Omega(j)$ on the prioritized dimension of an embedding 
from $\ell_p$ into $\ell_\infty$, with distortion smaller than $\sqrt{2}$.
Next, for distortion $1+\eps$ we prove a stronger lower bound with exponential dependency on $\eps$.

\begin{theorem}\label{thm:L2toLinftyPriorLB}
	For every $\eps\in(0,1)$ and $p\in[1,\infty]$ there is a set of points in $\ell_p$ and a priority ordering, 
such that every embedding of them into $\ell_\infty$ with distortion $1+\eps$ has prioritized dimension at least $j^{\frac{1}{6\eps}}$.
\end{theorem}
\begin{proof}
	As above, we may assume that $p=2$. Furthermore, we will assume that $\eps<\frac16$, as otherwise a better lower bound follows from \Cref{thm:LBL2toLinfty}.
	Let $n$ be large enough, and $H_n=\{\pm1\}^n\subseteq\ell_2^n$ be the Hamming cube. 
	We start by creating a symmetric subset $A\subset H_n$ (i.e. $A=-A$), such that all the points in $A$ differ in more than $\eps'n$ coordinates, for $\eps'=3\eps$.
	The set $A$ is created in a greedy manner. Initially set $S=H_n$ and $A=\emptyset$. First pick an arbitrary pair $x,-x\in S$ from $S$ and add them to $A$. 
Delete from $S$ all the points that differ in fewer than $\eps'\cdot n$ coordinates from either $x$ or $-x$. Note that when $y\in S$ is deleted, so is its 
\emph{antipodal} point $-y$. Thus, both $S,A$ are maintained to be symmetric.
	We continue with this process until $S$ is empty.
	The number of points that differ by at most $\eps'n$ coordinates from any point $v\in H$ is $\sum_{i=0}^{\eps'n}{n \choose i}\le {n\choose 
\eps'n}(1+\frac{\eps'n}{n-2\eps'n+1})< 2{n\choose \eps'n}$. Therefore for each added vertex we deleted fewer than $2{n\choose \eps'n}$ points. We conclude that the size of $A$ is at least
	\begin{eqnarray}
	|A|\ge\frac{2^{n}}{2\cdot{n \choose \eps'n}}\ge\frac{1}{2}\cdot\frac{2^{n}}{\left(\frac{en}{\eps'n}\right)^{\eps'n}}=\frac{1}{2}\cdot2^{\left(1-\eps'\log\frac{e}{\eps'}\right)n}>2\cdot2^{\frac n2}~.\label{eq:AsizeLB}
	\end{eqnarray}
	We argue that an embedding $f$ of $A$ into $\R$ can satisfy at most a single antipodal pair $x,-x$. 
	Indeed, assume by way of contradiction that there is $f:A\rightarrow\R$ and $x,y\in A$ such that 
	$\sqrt{4n}\le|f(x)-f(-x)|,|f(y)-f(-y)|\le\left(1+\eps\right)\sqrt{4n}$.
	Similar to the proof of \Cref{thm:LBL2toLinfty}, by case analysis, there must be a pair $z\in\{x,-x\}$ and $w\in\{y,-y\}$ such that
	$\left|f(z)-f(w)\right|\ge\min\left\{\left|f(x)-f(-x)\right|,\left|f(y)-f(-y)\right|\right\}\ge\sqrt{4n}$.
	As both $x,-x$ differs from both $y,-y$ in more than $\eps'n$ coordinates,  $z$ coincides with $w$ in at least $\eps'n$ 
coordinates. In particular $\|z-w\|_{2}\le\sqrt{(1-\eps')\cdot4n}$ .
	Thus $f$ has distortion at least $\frac{\left|f(z)-f(w)\right|}{\|z-w\|_{2}}\ge\frac{\sqrt{4n}}{\sqrt{(1-\eps')\cdot4n}}>1+\eps$, a contradiction.

	Next, let $Y=\{\pm1\}^{\eps'n}\{0\}^{(1-\eps')n}$ be the set of all points that attain values $\{\pm1\}$ in the first $\eps'n$ coordinates, with all other coordinates 0.
	Consider a coordinate $f:X\rightarrow \R$ that sends all of $Y$ to $\vec{0}$. We argue that $f$ will not satisfy any antipodal pair in $A$. Indeed, consider an antipodal pair $x,-x$. Let $y\in Y$ be the point agreeing with $x$ on the first $\eps' n$ coordinates and $0$ everywhere else. 
	It holds that
	\begin{align*}
	\left|f(x)-f(-x)\right| & \le\left|f(x)-f(y)\right|+\left|f(y)-f(-y)\right|+\left|f(-y)-f(-x)\right|\\
	& \le\left(1+\eps\right)\left(\left\Vert x-y\right\Vert _{2}+0+\left\Vert (-x)-(-y)\right\Vert _{2}\right)\\
	& =\left(1+\eps\right)\cdot2\cdot\sqrt{\left(1-\eps'\right)n}<\sqrt{4n}~.
	\end{align*}
	As each coordinate can satisfy at most a single antipodal pair from $A$, we conclude that every $1+\eps$ embedding of $X$ into $\ell_\infty$ must be non-zero 
on $Y$ in at least $|A|/2$ coordinates.

	We can now conclude the proof: Assume by way of contradiction that for any set in $\ell_2$ there is a $1+\eps$ embedding into $\ell_\infty$ with prioritized dimension $j^{\frac{1}{6\eps}}$.
	Set priority $\pi$ for $X = A \cup Y$ with the points in $Y$ occupying the first $|Y|$ places. 
By our assumption, there is a $1+\eps$ embedding where the points of $Y$ are non-zero only in the first 
	\[
	|Y|^{\frac{1}{6\eps}}=\left(2^{\eps'n}\right)^{\frac{1}{2\eps'}}=2^{\frac{n}{2}}\overset{(\ref{eq:AsizeLB})}{<}\frac{|A|}{2}
	\]
	coordinates. Thus the embedding cannot satisfy all the pairs in $A$, a contradiction.
\end{proof}

\section{Trees}
In this section, we present an embedding of trees into $\ell_\infty$ with  prioritized dimension $O(\log j)$.
We begin by sketching the classic isometric embedding of trees into $\ell^{O(\log n)}_\infty$ due to \cite{LLR95}.
This embedding utilizes a balanced decomposition (a technique also used in distance labelings for trees):
First, identify a separator vertex $s$ among the vertex set $V$, such that we can decompose $T$ into two trees $T_1,T_2$, each containing at most $\frac23n+1$ vertices, where $T_1\cap T_2=\{s\}$.
Now create a new coordinate wherein each vertex $v\in T_1$ assumes value $d(v,s)$,
while each vertex $x\in T_2$ assumes value $-d(x,s)$. This coordinate satisfies all pairwise distances $T_1\times T_2$.
Recursively (and separately) embed $T_1$ and $T_2$ into $\ell_\infty$, recalling that each has its own copy of $s$.
The two embeddings are then merged by translating $T_2$ so that its copy of $s$ is mapped to the same vector assumed by
the copy of $s$ in $T_1$.

Given a priority ordering on the vertices $v_{1},v_{2},\dots,v_{n}$,
our goal is to create an isometric embedding into $\ell_{\infty}$ with prioritized
dimension $O(\log j)$.
A natural first step would be to devise a terminal embedding:
Given terminal set $K$, embed $T$ into $\ell_\infty^{O(\log |K|)}$ while preserving all pairwise distances $K\times V$.
A terminal embedding can be constructed following the lines of the classic embedding by modifying the separator decision rule,
and ensuring that after $O(\log |K|)$ recursive steps each terminal is found in a different subtree.
However, a terminal embedding of this type is too weak to yield a prioritized embedding, since the mapping of all terminals into $\vec{0}$
(subsequent to their first $O(\log k)$ non-zero coordinates)
interferes with the distances between non-terminal pairs.

To circumvent this problem, we shall ``fold'' the terminals one above the other,
until ultimately all terminals will fall on a single representative vertex (see \Cref{lem:terminalTree}).
During such a folding, some of the non-terminal vertices will fold upon each other as
well, but our terminal embedding will be sufficiently robust to ensure that their
distances are retained. We will then use this result on terminal embeddings of trees
into $\ell_\infty$ (Lemma \ref{lem:terminalTree}) to derive the stronger result,
priority embeddings of trees into $\ell_\infty$ (Theorem \ref{thm:TreePriorLinfty}).

\subsection{Terminal Lemma}
\begin{lemma}\label{lem:terminalTree}
        Given a weighted tree $T=(V,E,w)$ and a
        set $K$ of $k$ terminals, there exist a pair of embeddings $f:T\rightarrow\ell_{\infty}^{O(\log k)}$
        and $g:T\rightarrow\mathcal{T}$ (into another weighted tree $\mathcal{T}$) such that the following
        properties hold:
        \begin{enumerate}
                \item[(1)] Lipschitz: For every $x,y\in V$, $\|f(x)-f(y)\|_{\infty}\le d_{T}(x,y)$
                and $d_{\mathcal{T}}(g(x),g(y))\le d_{T}(x,y)$.
                \item[(2)] Preservation: For every $x,y\in V$, either $\|f(x)-f(y)\|_{\infty}=d_{T}(x,y)$
                or $d_{\mathcal{T}}(g(x),g(y))=d_{T}(x,y)$, or both.
                \item[(3)] Terminal Collapse: $g$ maps all of $K$ into a single vertex, i.e.\
                $|g(K)|=1$.
        \end{enumerate}
\end{lemma}
\begin{proof}
We may assume that all terminals of $K$ are leafs, as otherwise we can simply add a dummy vertex in place of each terminal,
and connect the terminal to the dummy vertex with an edge of weight $0$.
        The proof is by induction on $k$.

\paragraph{Base cases.} For the case $k=1$ we can just return the tree as is, along with the null embedding into $\ell_{\infty}$.
        Next we prove the case of $k=2$.
        Denote the two terminals by $t_{1},t_2$,
and let $P$ be the unique path in $T$ connecting $t_{1},t_2$.
Let $c \in V$ be the midpoint of $t_{1}$ and $t_2$, such that $d_{T}(t_1,c)=d_{T}(t_2,c)$.
(If $c$ does not exist in $V$, then add $c$ to $V$, and split the corresponding middle edge into two new edges joined at $c$.)
Now ``fold'' $P$ around $c$. That is, create a new tree $\mathcal{T}$,
where path $P$ is replaced by a new path that ends at $c$, and every $x \in P$ is found on the new path at distance exactly
$d_T(x,c)$ from $c$. Any pair of points in $P$ equidistant from $c$ are merged -- and in particular $t_1$ and $t_2$ are now
the same point, which is the other endpoint of the new path. All the other edges and vertices remain the same. As a result, we obtain an embedding $g:d_{T}\rightarrow\mathcal{T}$ (see \Cref{fig:folding} for an illustration).
It is clear that $g$ is Lipschitz, and moreover $|g(\{t_1,t_2\})|=1$.

        \begin{figure}[t]
                \centering{\includegraphics[scale=0.85]{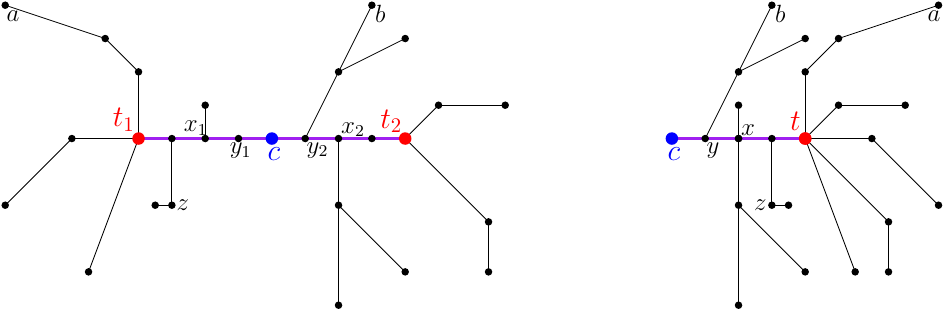}}
                \caption{\label{fig:folding}\small \it
                        On the left is illustrated the tree $T$ with two terminals $t_1$,$t_2$. The path $P$ between the terminals is colored in purple. The
(possibly imaginary) vertex $c$ lies at the midpoint of $t_1$ and $t_2$.
                        On the right is illustrated the tree $\mathcal{T}$ which is obtained by ``folding'' the path $P$ around $c$. \newline
                        In this example, all the edges in $T$ are of unit weight, except for the edge $\{y_1,y_2\}$ that has weight $2$.
The values of the function $f:T\rightarrow\R$ (eq. (\ref{eq:defFoldingFunc})) are: $f(t_1)=4, f(t_2)=-4, f(a)=7, f(b)=-3, f(x_1)=2, f(x_2)=-2, f(z)=4$.}
        \end{figure}

Having specified the function $g$, we now describe the function $f$:
separate $T$ into two trees $T_1$,$T_2$ where $T_1\cap T_2=\{c\}$. 
Set the function $f:V\rightarrow \mathbb{R}$ as follows.
        \begin{eqnarray}
        f(v)=\begin{cases}
        d_{T}(v,c) & v \in T_1\\

        -d_{T}(v,c) & v \in T_2\setminus\{c\}
        \end{cases}~~~.
        \label{eq:defFoldingFunc}
        \end{eqnarray}
        See \Cref{fig:folding} for an illustration of function $f$. We argue that $f$ is Lipschitz:
Consider a pair of vertices $u,v$. If $u,v \in T_i$ (for some $i$), then by the triangle inequality
        $|f(u)-f(v)|=|d_{T}(u,c)-d_{T}(v,c)|\le d_{T}(u,v)$.
        Otherwise, assume without loss of generality that $u \in T_1$ while $v \in T_2$.
The shortest path from $u$ to $v$ must pass through $c$, thus
        \begin{eqnarray}
        |f(u)-f(v)|=|d_{T}(u,c)+d_{T}(v,c)|=d_{T}(u,v)~.
        \label{eq:vPuPdiff}
        \end{eqnarray}
        It remains only to prove the second property (preservation).
Consider a pair of vertices $u,v$. If $u \in T_1$ and $v \in T_2$, then by equation (\ref{eq:vPuPdiff}) $|f(u)-f(v)|=d_{T}(u,v)$.
Otherwise, if $u,v \in T_i$, the shortest path between $u$ and $v$ in $T$ is isomorphic to the shortest path in $\mathcal{T}$,
and so $d_{\mathcal{T}}(u,v)=d_{T}(u,v)$ as required.

\paragraph{Induction step.}
For $k>2$ terminals, we will assume by induction that for every tree with $k'<k$
terminals there are embeddings $f,g$ as required above, such that $f$ uses at most $a\cdot\log k'$ coordinates, for $a=\frac{2}{\log(3/2)}$.
Consider a tree $T$, and a terminal set $K$ of size $k$.
Let $s\in V$ be a separator vertex, such that $T$ can be separated into two trees $T_1$,$T_2$ where $T_1\cap T_2=\{s\}$, and each $T_i$ contains at most $\frac{2}{3}k$ terminals.
        As all the terminals are leafs, $s\notin K$.
        Create a single new coordinate $h^{s}:V\rightarrow\mathbb{R}$ defined as follows
        \[
        h^{s}(x)=\begin{cases}
        d_{T}(x,s) & x\in T_1\\
        -d_{T}(x,s) & x\in T_2
        \end{cases}
        ~~~.
        \]
        It is clear that $h^{s}$ is Lipschitz, and that for every $x\in T_1,y\in T_{2}$,~  $|h^{s}(x)-h^{s}(y)|=d_{T}(x,y)$.
        For $i\in\{1,2\}$, invoke the induction hypothesis on $T_{i}$ with terminal set $K_i=T_{i}\cap K$,
creating embedding pair $f_{i}:T\rightarrow\ell_\infty^{a\cdot\log(|K_i|)}$ and $g_{i}:T\rightarrow\mathcal{T}_{i}$
which together satisfy requirements (1)-(3).
By padding with $0$-valued coordinates, we can assume that both $f_{1}$ and $f_{2}$ use exactly $a\cdot\log\frac{2}{3}k$ coordinates.
Moreover, by translation we can assume that $f_{1}(s)=f_{2}(s)=\vec{0}$ (note that there are no prioritized/terminal dimension guarantees here). Set
        $f_{12}$ to be the combined function of $f_{1},f_{2}$:
        \[
        f_{12}(x)=\begin{cases}
        f_{1}(x) & x\in T_1\\
        f_{2}(x) & x\in T_2
        \end{cases}~~~.
        \]
        We argue that the function $f_{12}$ is Lipschitz: For $x,y\in T_{i}$,
        $\|f_{12}(x)-f_{12}(y)\|_{\infty}=\|f_{i}(x)-f_{i}(y)\|_{\infty}\le d_{T_i}(x,y)=d_{T}(x,y)$.
On the other hand for $x\in T_1$ any $y\in T_{2}$, using the triangle inequality
        \begin{align*}
        \|f_{12}(x)-f_{12}(y)\|_{\infty} & \le\|f_{12}(x)-f_{12}(s)\|_{\infty}+\|f_{12}(s)-f_{12}(y)\|_{\infty}\\
        & \le d_{T_{1}}(x,s)+d_{T_{2}}(s,y)=d_{T}(x,s)+d_{T}(s,y)=d_{T}(x,y)~.
        \end{align*}
        Set $f_{12s}$ to be the concatenation of $f_{12}$ with $h^s$, and it is clear that $f_{12s}$ is Lipschitz as well.
This completes the description of the embedding into $\ell_\infty$.

For the embedding into the tree,
let $\mathcal{T}_{12}$ be composed of the trees $\mathcal{T}_{1}$ and $\mathcal{T}_{2}$ glued together in $g_{1}(s),g_{2}(s)$.
Similarly define
$g_{12}:T\rightarrow\mathcal{T}_{12}$ as follows
        \[
        g_{12}(x)=\begin{cases}
        g_{1}(x) & x\in T_{1}\\
        g_2(x) & x\in T_{2}
        \end{cases}~~~.
        \]
        Using the triangle inequality in the same manner as for $f_{12}$,
        it is clear that $g_{12}$ is Lipschitz.

        We argue that requirement (2) holds w.r.t.\ $f_{12s},g_{12}$. 
Indeed, for $u,v$ in $T_i$,
        \begin{align*}
       \max\left\{ \|f_{12s}(x)-f_{12s}(y)\|_{\infty},d_{\mathcal{T}_{12}}(g_{12}(x),g_{12}(y))\right\}
& \ge\max\left\{ \|f_{i}(x)-f_{i}(y)\|_{\infty},d_{T_{i}}(g_{i}(x),g_{i}(y))\right\} \\
        & =d_{T_{i}}(x,y)=d_{T}(x,y)
        \end{align*}
        On the other hand, for $v\in T_1,~u\in T_2$,
        \[
        \max\left\{ \|f_{12s}(v)-f_{12s}(u)\|_{\infty},d_{\mathcal{T}_{12}}(g_{12}(v),g_{12}(u))\right\} \ge|h^{s}(v)-h^{s}(u)|=d_{T}(v,u)~.
        \]

        However, requirement (3) does not immediately hold, as $\mathcal{T}_{12}$ contains two terminals $g_1(K_1),g_2(K_2)$.
Invoke the lemma for the case of $k=2$ to create
two embeddings $\hat{f}:\mathcal{T}_{12}\rightarrow \mathbb{R},~\hat{g}:\mathcal{T}_{12}\rightarrow \mathcal{T}$ that fulfill requirements $(1)-(3)$.
Set $f=f_{12s}\oplus \hat{f}(g_{12})$ to be the concatenation of $f_{12s}$ with $\hat{f}(g_{12})$ and $g=\hat{g}(g_{12})$ to be the
composition of $\hat{g}$ with $g_{12}$ ending in the tree $\mathcal{T}$.
        It is clear that both $f,g$ are Lipschitz, as the Lipschitz 
property is preserved under concatenation and composition, thus 
establishing (1). Moreover, $g$ maps all terminals to a single vertex.
        Requirement (2) also holds:
        \begin{align*}
        d_{T}(u,v) & =\max\left\{ \|f_{12s}(v)-f_{12s}(u)\|_{\infty},d_{\mathcal{T}_{12}}(g_{12}(v),g_{12}(u))\right\} \\
        & =\max\left\{ \|f_{12s}(v)-f_{12s}(u)\|_{\infty},\left|\hat{f}(g_{12}(v))-\hat{f}(g_{12}(u))\right|,
d_{\mathcal{T}}\left(\hat{g}(g_{12}(v)),\hat{g}(g_{12}(v))\right)\right\} \\
        & =\max\left\{ \|f(v)-f(u)\|_{\infty},d_{\mathcal{T}}(g(v),g(v))\right\}~.
        \end{align*}
        Finally, and recalling that $a=\frac{2}{\log(3/2)}$, the number of coordinates is bounded by
        \[
        a\cdot\log\frac{2}{3}k+1+1=a\cdot\log k+\left(a\cdot\log\frac{2}{3}+2\right)=a\cdot\log k~.
        \]
        The lemma now follows.
\end{proof}

\subsection{Prioritized Embedding of Trees into $\ell_\infty$}
\begin{theorem}\label{thm:TreePriorLinfty}
	Given a weighted tree $T=(V,K,w)$ and a priority ordering $\pi$ over $V$,
	there is an isometric embedding $f$ into $\ell_{\infty}$ with prioritized
	dimension $O(\log j)$.
\end{theorem}
\begin{proof}
	\sloppy Let $\pi=\left\{ x_{1},x_{2},\dots,x_{n}\right\} $ be a 
priority order. Set $S_{i}=\left\{ x_{i} : i\le2^{2^{i}}\right\} $ for 
$1\le i\le\left\lceil \log\log n\right\rceil $.
	Using \Cref{lem:terminalTree}, w.r.t.\ terminal set $S_{1}$
	construct embeddings $f_{1}:T\rightarrow\ell_{\infty}^{O\left(\log\left|S_{1}\right|\right)}$ and $g_{1}:T\rightarrow T_{1}$. It holds that $g_1(S_1)$ is a single vertex in $T_1$, and for every $u,v\in V$, $d_{T}(u,v)=\max\left\{ \left\Vert f_{1}(u)-f_{1}(v)\right\Vert _{\infty},d_{T_{1}}\left(g_{1}(u),g_{1}(v)\right)\right\} $.
	Next, using \Cref{lem:terminalTree} again, w.r.t.\ terminal set $g_{1}\left(S_{2}\right)$, construct embeddings  
$f_{2}:g_{1}\left(T\right)\rightarrow\ell_{\infty}^{O\left(\log\left|S_{2}\right|\right)}$ and $g_{2}:g_{1}\left(T\right)\rightarrow T_{2}$.
	By translation, we can assume that $f_{2}(g_{1}(S_{1}))=\vec{0}$. Furthermore, $g_2(g_1(S_2))$ is a single vertex in $T_2$. It also holds that,
	\[
	d_{T}(u,v)=\max\big\{ \left\Vert f_{1}(u)-f_{1}(v)\right\Vert_{\infty},\left\Vert f_{2}(g_{1}(u))-f_{2}(g_{1}(v))\right\Vert _{\infty},d_{T_{2}}\left(g_{2}(g_{1}(u)),g_{2}(g_{1}(v))\right)\big\}. 
	\]
	More generally, in the $i$-th step, we invoke 
\Cref{lem:terminalTree} on $T_{i-1}$ (w.r.t.\ terminal set 
$g_{i-1}(g_{i-2}(\cdots(g_1(S_i))))$ ) to construct tree $T_i$ and embeddings $f_i,g_i$.
	By induction, we constructed trees $T_{1},\dots,T_{i}$ and embeddings
	$f_1:T\rightarrow\ell_\infty^{O(\log|S_1|)},\dots,f_i:T_{i-1}\rightarrow\ell_\infty^{O(\log|S_i|)}$,~ $g_1:T\rightarrow T_1,\dots,g_i:T_{i-1}\rightarrow T_i$ such that for all $q\in[1,i]$, $g_q(g_{q-1}(\dots(g_1(S_q))))$ is a single vertex in $T_q$ and $f_q(g_{q-1}(\dots(g_1(S_{q-1}))))=\{\vec{0}\}$. Furthermore
	\begin{eqnarray}\label{eq:maxNormsTree}
	d_{T}(u,v) & =&\max\big\{\left\Vert f_{1}(u)-f_{1}(v)\right\Vert _{\infty},\dots,\left\Vert f_{i}(g_{i-1}(\dots(g_{1}(u))))-f_{i}(g_{i-1}(\dots(g_{1}(u))))\right\Vert _{\infty}\\
	& &\qquad~~,d_{T_{i}}\big(g_{i}(g_{i-1}(\dots(g_{1}(u))))\ ,\ g_{i}(g_{i-1}(\dots(g_{1}(u))))\big)\big\}~.\nonumber
	\end{eqnarray}
	Denote $\alpha=\lceil\log\log n\rceil$. After $\alpha$ steps we get functions and trees as above. 
	Set $f=f_{1}\oplus\left(f_{2}\circ g_{1}\right)\oplus\left(f_{3}\circ g_{2}\circ g_{1}\right)\oplus\cdots\oplus\left(f_{\alpha}\circ g_{\alpha-1}\circ\cdots\circ g_{1}\right):T\rightarrow\ell_{\infty}$.
	We argue that $f$ is an isomorphic embedding with prioritized dimension $O\left(\log j\right)$ as promised.
	Note that all vertices of $V$ belong to $S_\alpha$ and hence  mapped by $g_\alpha(g_{\alpha-1}(\cdots(g_1)))$ to the 
same vertex. Thus for every $u,v\in V$, $d_{T_\alpha}\big(g_\alpha(g_{\alpha-1}(\cdots(g_1(u)))),g_\alpha(g_{\alpha-1}(\cdots(g_1(v))))\big)=0$. By equation (\ref{eq:maxNormsTree}) we get
	\begin{align*}
	d_{T}(u,v) & =\max\left\{ \left\Vert f_{1}(u)-f_{1}(v)\right\Vert _{\infty},\dots,\left\Vert f_{\alpha}(g_{\alpha-1}(\dots(g_{1}(u))))-f_{\alpha}(g_{\alpha-1}(\dots(g_{1}(u))))\right\Vert _{\infty}\right\} \\
	& =\left\Vert f(u)-f(v)\right\Vert _{\infty}.
	\end{align*}
	Finally we argue that $f$ has prioritized dimension $O(\log j)$. Consider $x_{j}\in S_{\left\lceil \log\log j\right\rceil }$.
	For every $i>\left\lceil \log\log j\right\rceil $ it holds that
	$f_{i}\left(g_{i-1}\left(g_{i-2}\left(\cdots\left(g_{1}(x_{j})\right)\right)\right)\right)=\vec{0}$
	(as $x_{j}\in S_{i-1}$). Therefore $x_{j}$ might be non-zero only in the first
	\[
	\sum_{i=1}^{\left\lceil \log\log j\right\rceil }O\left(\log|S_{i}|\right)=O\left(\sum_{i=1}^{\left\lceil \log\log j\right\rceil }2^{i}\right)=O\left(2^{\left\lceil \log\log j\right\rceil +1}\right)=O\left(\log j\right)
	\]
	coordinates.	
\end{proof}

\section{Planar Graphs}
The theorem below demonstrates that any isometric embedding of the cycle graph $C_{2n}$ into $\ell_\infty$ requires dimension $n$. 
Furthermore, no prioritized dimension is possible for isometric embeddings of the cycle graph. 
The cycle graph is an interesting example as it is both planar and has treewidth $2$. 
The non-prioritized lower bound is a special case of a theorem proved in \cite{LLR95}, which applies to general norms. 
Nonetheless, the proof provided here is much simpler.
\begin{theorem}\label{thm:PlanarEmbeddingLB}
	For every $n\in N$, every isometric embedding of $C_{2n}$ (the unweighted cycle graph) into $\ell_\infty$ requires at least $n$ coordinates.
	Furthermore, there is no function $\alpha:\N\rightarrow\N$ for which the family of cycle graphs $\{C_n\}_{n\in \N}$ can be
	embedded into $\ell_{\infty}$ with prioritized dimension $\alpha$.
\end{theorem}
\begin{proof}
	Denote the vertices of $C_{2n}$ by $V=\{v_0,v_1,\dots,v_{2n-1}\}$. 
The maximum distance is $n$, and it is realized on all the antipodal pairs $\{v_0,v_n\},\{v_{1},v_{n+1}\},\dots,\{v_{n-1},v_{2n-1}\}$. 
We argue that in a single embedding into the line $\R$, at most one antipodal pair can be satisfied, that is realize distance $n$. 
Indeed, suppose by way of contradiction that there is a non-expansive function $f:C_{2n}\rightarrow\R$ such that $|f(v_j)-f(v_{n+j})|=|f(v_i)-f(v_{n+i})|=n$ for $i\ne j$, 
then necessarily \sloppy $\max\left\{|f(v_i)-f(v_{j})|,|f(v_i)-f(v_{n+j})|,|f(v_{n+i})-f(v_{j})|,|f(v_{n+i})-f(v_{n+j})|\right\} \ge n$, a contradiction.
	As there are $n$ antipodal pairs, every isometric embedding requires at least $n$ coordinates. 
	
	For the second part, for sufficiently large $n$ set a priority ordering $\pi$ of $C_{2n}$ where $v_n,v_{n+1}$ have priorities $1$ and $2$ respectively. 
Consider a single Lipschitz coordinate $f:C_{2n}\rightarrow\R$ sending both $v_n,v_{n+1}$ to $0$. By the triangle inequality, for every antipodal pair $\{v_i,v_{n+i}\}$, it holds that
	\begin{align*}
	|f(v_{i})-f(v_{n+i})| & \le|f(v_{i})-f(v_{n})|+|f(v_{n})-f(v_{n+1})|+|f(v_{n+1})-f(v_{n+i})|\\
	& \le(n-i)+0+(i-1)=n-1<n~.
	\end{align*}
	Thus no antipodal pair could be satisfied. We conclude that $f(v_{i})\ne f(v_{n+i})$ in at least $n$ coordinates. In particular for $\alpha(2)<n$, priority distortion $\alpha$ is impossible.
\end{proof}

\begin{theorem}\label{thm:PlanarPrioritzedLB}
	Every isometric prioritized labeling scheme for planar graphs must have prioritized label size of at least $\Omega(j)$ (in bits).
	This lower bound holds even for unweighted planar graphs.
\end{theorem}
\begin{proof}
	Recall that \cite{GPPR04} proved an $\Omega(n^{\frac13})$ lower bound on the label size for exact distance labeling for unweighted planar graphs.
	We will use the same example graph $G$ from \cite{GPPR04}. We refer to \cite{GPPR04} for the description of $G$; here it suffices to describe its relevant properties.
	Given a parameter $n$, $G=(V,E)$ is an unweighted planar graph with $O(n^3)$ vertices, among which $O(n)$ lie on the outer face, denoted $\tilde{V}\subset V$. 
Set $E=E_1\cup E_2$, where $|E_2|=\Omega(n^2)$. 
For every subset $A\subset E_2$, denote by $G_{A}=(V,E_1\cup A)$ the graph $G$ wherein the edge-set $E_2\setminus A$ has been removed 
(equivalently, where only the edge-set $E_1\cup A$ is retained).
	\cite{GPPR04} showed that given all pairwise distances between the outerface vertices  $\{d_{G_A}(v,u)\mid v,u\in \tilde{V}\}$, one can recover the set $A$.
	Note that $\log  2^{|E_2|}=\Omega(n^2)$ bits are required to encode the set $A$.
	
	Suppose by way of contradiction that there is an exact prioritized labeling scheme with $o(j)$ labels size (in bits). 
Given a graph $G_{A}$, we define a priority ordering where the vertices of $\tilde{V}$ occupy the first $|\tilde{V}|$ places. 
Given all the labels of $\tilde{V}$, we can encode the set $A$ by simply concatenating all the labels. 
Therefore the sum of the lengths of the labels of $\tilde{V}$ must be $\Omega(n^2)$. 
However, by our assumption, the sum of their lengths is only $\sum_{j=1}^{|\tilde{V}|}o(j)=o(n^2)$, a contradiction 
(for sufficiently large $n$).	
\end{proof}

\section{Conclusions and Open Questions}\label{sec:OpenQues}
We uncover a wide spectrum of settings and bounds that answer our questions.
For \Cref{q1:LabVsEmb},
in the simplest case of trees, labeling and embeddings have similar behavior, and both admit prioritization with similar bounds.
For the least restricted case of general graphs/metrics,
we find similarly that labelings and embeddings exhibit similar behavior across various distortion parameters.
However between these two extremes, for $\ell_p$ spaces, planar graphs and treewidth $k$ graphs, we see significant separations between labelings and embeddings.

For \Cref{q2:PriorVs}, we show that labelings admit far superior prioritized versions than their embedding counterparts in all settings other than trees,
and most notably for general graphs and for planar/bounded-treewidth graphs,
where no prioritized dimension is possible. In $\ell_p$ spaces, while we did not rule out the possibility of prioritized dimension, we demonstrate a surprising exponential gap between labelings and embeddings (also in the dependence on $\eps$).

For \Cref{q3:priorVersion} we saw that labeling schemes have prioritized versions, in all cases other than planar graphs where instead of the desired $O(\sqrt{j})$ label size we show that $\Theta(j)$ is surprisingly necessary.
For embeddings into $\ell_\infty$ we showed that for larger distortion some prioritized dimension is possible, even though it is much worse than its labeling counterpart. 

Our results leave a few open questions that may be of independent interest:
\begin{enumerate}
	\item How many coordinates are required in order to embed planar graphs -- or even treewidth $2$ graphs -- into $\ell_\infty$ with distortion $1+\eps$?
	\item What is the required label size for $1+\eps$ distance labeling for $\ell_p$ spaces, for $p>2$?
	\item Is it possible to embed $\ell_p$ spaces ($p \in [1,\infty]$) into $\ell_\infty$ with distortion $1+\eps$ and some prioritized dimension? 
\Cref{thm:L2toLinftyPriorLB}  provided a $j^{\Omega(\frac1\eps)}$ lower bound, but did not rule out this possibility. 
The same question applies when considering constant distortion.
	\item All results on embedding of general graphs into $\ell_\infty$ with both prioritized distortion and dimension 
(our \Cref{thm:GeneralGraphPriorityEmbedding}, Theorem 15 in \cite{EFN18}, 
and Theorems 2,3 in \cite{EN20}) feature prioritized {\em contractive} distortion. What is possible w.r.t.\ 
classic prioritzed distortion (see footnote (\ref{foot:PrioritzedContractive}))? 
\end{enumerate}

{\small
\bibliographystyle{alphaurlinit}
\bibliography{Bib}
}
\end{document}